\documentclass[12pt]{amsart}
\author{Simeon Ball, Guillermo Gamboa and Michel Lavrauw}\thanks{
The first author acknowledges the support of the project  MTM2017-82166-P of the Spanish {\em Ministerio de Ciencia y Innovaci\'on.}}

\usepackage{amsmath}
\usepackage{amssymb}
\usepackage{amsthm}
\usepackage{graphicx}
\usepackage{amscd}
\usepackage{amsthm,amsmath,amssymb,amsfonts,amscd}
\usepackage{physics}
\usepackage{enumerate}
\usepackage{epstopdf}
\usepackage{booktabs}
\usepackage{subcaption}
\usepackage{array}
\usepackage{multirow}
\usepackage{listings}
\usepackage{url}

\newtheorem{theorem}{Theorem}
\newtheorem{lemma}[theorem]{Lemma}

\title{{\rm On additive MDS codes over small fields}}

\usepackage{amsmath}
\usepackage{amssymb}
\usepackage{amsthm}
\usepackage{graphicx}
\usepackage{amscd}

\begin{document}
\baselineskip=17pt
\date{11 December 2020}
\maketitle
\begin{abstract}
Let $C$ be a $(n,q^{2k},n-k+1)_{q^2}$ additive MDS code which is linear over ${\mathbb F}_q$. We prove that if $n \geqslant q+k$ and $k+1$ of the projections of $C$ are linear over ${\mathbb F}_{q^2}$ then $C$ is linear over ${\mathbb F}_{q^2}$. We use this geometrical theorem, other geometric arguments and some computations to classify all additive MDS codes over ${\mathbb F}_q$ for $q \in \{4,8,9\}$. We also classify the longest additive MDS codes over ${\mathbb F}_{16}$ which are linear over ${\mathbb F}_4$. In these cases, the classifications not only verify the MDS conjecture for additive codes, but also confirm there are no additive non-linear MDS codes which perform as well as their linear counterparts. These results imply that the quantum MDS conjecture holds for $q \in \{ 2,3\}$.
\end{abstract}
\section{Introduction}

Let $A$ be a finite set and let $n$ and $k$ be positive integers. An {\em MDS code} $C$ is a subset of $A^n$ of size $|A|^k$ in which any two elements of $C$ differ in at least $n-k+1$ coordinates. In other words, the minimum (Hamming) distance $d$ between any two elements of $C$ is $n-k+1$. In general, we denote a code $C \subseteq A^n$ with minimum distance $d$ as a $(n,|C|,d)_{|A|}$ code. If there is no restriction on the size of $A$ then MDS codes are the best performing codes when we apply nearest neighbour decoding. They have the property that a codeword can be recovered from any $k$ coordinates, which makes them very useful, for example, in distributed storage systems.

The ubiquitous example of an MDS code is the Reed-Solomon code. The Reed-Solomon code is an example of a linear code in which the alphabet is a finite field ${\mathbb F}_q$ and $C$ is a $k$-dimensional subspace of ${\mathbb F}_q^n$. The Reed-Solomon code has length $n=q+1$ which can be extended to a code of length $q+2$ in the case that $k \in \{3,q-1\}$ and $q$ is even. Its codewords are the evaluation of polynomials of degree at most $k-1$. To give a more precise definition, suppose ${\mathbb F}_q=\{a_1,\ldots,a_q \}$. The Reed-Solomon code is
$$
C=\{(f(a_1),\ldots,f(a_{q}), c_f) \ | \ f \in {\mathbb F}_q[X],\ \ \mathrm{deg}\  f \leqslant k-1\}, 
$$
where $c_f$ is coefficient of $X^{k-1}$ in $f$.

There are no known MDS codes which better the Reed-Solomon code and it is generally assumed that there are none. The {\em MDS conjecture} reflects this and states that for an $(n,q^k,d)_q$ MDS code where $d \geqslant 3$, the length $n$ satisfies $n \leqslant q+1$, unless $k \in \{3,q-1\}$ and $q=2^h$ in which case $n \leqslant q+2$. 

The MDS conjecture has been verified for linear codes when $q$ is prime \cite{Ball2012}. It is also known to hold for linear codes when $q$ is square and $k \leqslant c\sqrt{q}$, where the constant $c$ depends on whether $q$ is odd or even. And for $q$ non-square and $k \leqslant c'\sqrt{pq}$, where again the constant $c'$ depends on whether $q$ is an odd power of an even or odd prime. See \cite{BL2019} for a recent survey. It is also known to hold for all MDS codes over alphabets of size at most $8$, see \cite{KO2016}. Here we present some evidence that the MDS conjecture is true for additive MDS codes over finite fields by proving the conjecture for additive MDS codes over ${\mathbb F}_9$ and ${\mathbb F}_{16}$, where in the last case we assume linearity over ${\mathbb F}_4$. 

An MDS code over $A$ with $d=1$ is $A^k$ and with $k=1$ is the repetition code, so these are trivial. For $d=2$, an MDS code is equivalent to a Latin $k$-cube of order $|A|$, so we only consider MDS codes with $d \geqslant 3$. For alphabets of size $2$, there are no non-trivial MDS codes with $d \geqslant 3$ and for alphabets of size $3$, the only non-trivial MDS code with $d \geqslant 3$ is the unique $(4,3^2,3)_3$ code. Alderson \cite{Alderson2006} classified all MDS codes over alphabets of size $4$ by proving the uniqueness of the $(6,4^3,4)_4$ code and of the $(5,4^3,3)_4$ code. The non-existence of two mutually orthogonal Latin squares of order 6 implies the non-existence of nontrivial MDS codes with $d \geqslant 3$ over alphabets of size $6$. In the articles \cite{KDO2015} and \cite{KO2016}, all MDS codes over alphabets of size $5$, $7$ and $8$ are classified. It turns out that all MDS codes over alphabets of size $5$ and $7$ with $d\geqslant 3$, except the $(4, 7^2,3)_7$ codes, are equivalent to linear MDS codes. Here, for the sake of completeness, we also classify additive MDS codes over ${\mathbb F}_{8}$ and compare this classification to the results obtained in \cite{KO2016}. 

The fact that there are at most $q-1$ mutually orthogonal Latin squares of order $q$, implies that if there is an $(n,q^k,n-k+1)_q$ MDS code then $n \leqslant q+k-1$. This is known as the trivial upper bound and is due to Bush \cite{Bush1952}.

\section{Quantum MDS codes}

A {\em quantum code} on {\em $n$ subsystems} is a $K$-dimensional subspace of $({\mathbb C}^q)^{\otimes n}$. A code with minimum distance $d$ is able to detect errors, which act non-trivially on the code space, on up to $d-1$ of the subsystems and correct errors on up to $\frac{1}{2}(d-1)$ of the subsystems. If the dimension $K=q^k$ for some $k$ then we say the quantum code is an $[\![ n,k,d ]\!] _q$ code and if not simply an $ (\!( n,K,d)\!)_q$ code. 

To be able to describe in more detail quantum error-correcting codes which are of interest to us here, we 
specify the type of errors they can correct. Firstly, we suppose that $q$ is the power of a prime $p$. Let $\{ \ket{x} \ | \ x \in {\mathbb F}_q\}$ be a basis of ${\mathbb C}^q$. We define the following set of endomorphisms of ${\mathbb C}^q$ called the {\em Pauli operators}. For each $a,b\in {\mathbb F}_q$, we define $X(a)$ by its action on the basis vectors, $X(a)\ket{x}=\ket{x+a}$, and likewise $Z(b)$ by $Z(b)\ket{x}=e^{2\pi i \mathrm{tr}_{q \rightarrow p}(bx)/p} \ket{x}$, where 
$$
\mathrm{tr}_{q \rightarrow p}(x)=x+x^p+x^{p^2}+\cdots+x^{q/p}
$$
denotes the usual trace map from ${\mathbb F}_q$ to ${\mathbb F}_p$. The Pauli operators are of the form $X(a)Z(b)$, for some $a,b \in {\mathbb F}_q$. In the error model, the (Pauli) errors on $({\mathbb C}^q)^{\otimes n}$ are tensor products of Pauli operators. An error has weight $t$ if precisely $t$ of the components in the tensor product are not the identity operator, whilst the remaining $n-t$ are the identity operator.  A quantum error-correcting code of minimum distance $d$ is able to correct all Pauli errors of weight at most $\frac{1}{2}(d-1)$ which act non-trivially on the code subspace. Such quantum error-correcting codes are most commonly constructed by taking the joint eigenspace of eigenvalue $1$ of a subgroup of Pauli operators. These codes are called {\em stabiliser} codes.  See \cite{KKKS2006} for more details on stabiliser codes. The quantum Singleton bound states that for an $ [\![ n,k,d ]\!] _q$ quantum code, $k \leqslant n-2d+2$. A code attaining this bound is called a {\em quantum MDS code}. 

One of our motivations for studying additive MDS codes is the following theorem, Theorem~\ref{ketkarthm}. The notation $C^{\perp_a}$ is used to describe the orthogonal complement of $C$ with respect to the form, defined for $u,v \in {\mathbb F}_{q^2}^n$ by,  
$$
(u,v)_a=\mathrm{tr}_{q \rightarrow p} (\gamma( u\cdot v^q-u^q \cdot v)),
$$
for some $\gamma$ such that $\gamma^q=-\gamma$.

The following is \cite[Theorem 15]{KKKS2006} applied to MDS codes.

\begin{theorem} \label{ketkarthm}
A $[\![n,n-2(d-1),d]\!]_q$ stabiliser MDS code exists if and only if there is an additive $(n,q^{2(d-1)},n-d+2)_{q^2}$ MDS code $C$ such that $C \leqslant C^{\perp_a}$.
\end{theorem}
Thus, by ruling out the existence of additive MDS codes over ${\mathbb F}_{q^2}$, one can prove the non-existence of the corresponding stabiliser MDS code.

\section{Additive codes over a finite field}

Recall that we use the notation $(n,q^k,d)_{q}$ to denote a code of size $q^k$ which is a subset of $A^n$, where $|A|=q$ and in which the minimum distance is $d$. The parameter $n$ is the {\em length} of the code. If $q$ is a prime power and the code is linear over ${\mathbb F}_q$ then we say that the code is a $[n,k,d]_q$ code, in which case $k$ is the dimension of the code.

If $A$ is an abelian group then we define an additive code to be a code $C$ with the property that for all $u,v \in C$, we have $u+v \in C$.

\begin{theorem} \label{linearoverfp}
Let $q=p^h$ where $p$ is prime. An additive code $C \subseteq {\mathbb F}_q^n$ is a subspace over ${\mathbb F}_p$.
\end{theorem}

\begin{proof}
Let $u \in C$. Summing $n$ times the codeword $u$, we have that $nu \in C$, for all $n \in {\mathbb F}_p$. Since, by assumption, $u+v \in C$ for all $u,v \in C$, the code $C$ is a subspace over ${\mathbb F}_p$.
\end{proof}

We remark that if the field ${\mathbb F}_q$ has other proper subfields apart from ${\mathbb F}_p$ then there are interesting subsets of additive codes over ${\mathbb F}_q$ which are linear over the larger subfield. For example, we will be particularly interested in additive codes over ${\mathbb F}_{16}$, which are linear over ${\mathbb F}_4$, and have the parameters of a linear code over ${\mathbb F}_{16}$. For this reason, we shall fix the base field as ${\mathbb F}_q$ and consider additive $(n,q^{kh},d)_{q^h}$ codes which are linear over ${\mathbb F}_q$.

The {\em weight} of a vector $v$ is the number of non-zero coordinates that it has. The {\em minimum weight} of a code $C$ is the minimum weight of the non-zero vectors of $C$.

\begin{lemma} \label{minweight}
The minimum weight $w$ of an additive code $C \subseteq {\mathbb F}_{q^h}^n$ is its minimum distance $d$. 
\end{lemma}

\begin{proof}
Suppose that $q$ is the power of the prime $p$, that $u,v \in C$ and that the distance between $u$ and $v$ is $d$. Since $C$ is additive $(p-1)u=-u \in C$. Thus, $v-u \in C$. Since $v$ and $u$ differ in $d$ coordinates $v-u$ has weight $d$. Thus, $d\geqslant w$.

Suppose that $u \in C$ is a codeword of weight $w$. Since $0 \in C$, the distance between $u$ and $0$ is at least $d$, we have that $w \geqslant d$. 
\end{proof}

We denote by $\mathrm{PG}(k-1,q)$ the $(k-1)$-dimensional projective space over ${\mathbb F}_q$. This geometry has as points the one-dimensional subspaces of ${\mathbb F}_q^k$. For $i \in \{1,\ldots,k\}$, an $(i-1)$-dimensional subspace of $\mathrm{PG}(k-1,q)$ is given by an $i$-dimensional subspace $U$ of ${\mathbb F}_q^k$, and consists of the points whose corresponding one-dimensional subspace is contained in $U$. The action of the general linear group $\mathrm{GL}(k,q)$ on the points of $\mathrm{PG}(k-1,q)$  induces the projective general linear group $\mathrm{PGL}(k,q)$. This extends to $\mathrm{P}\Gamma\mathrm{L}(k,q) \cong \mathrm{PGL}(k,q) \rtimes \mathrm{Gal}({\mathbb F}_q / {\mathbb F}_p$) which is the full automorphism group of $\mathrm{PG}(k-1,q)$.

A {\em generator matrix} for a $(n,q^{hk},d)_{q^h}$ code $C$, which is linear over ${\mathbb F}_q$, is an $hk \times n$ matrix with entries from ${\mathbb F}_{q^h}$, whose row space over ${\mathbb F}_{q}$ is $C$. 

Let $\{ \epsilon_i \ | \ i=1,\ldots, h\}$ denote a basis for ${\mathbb F}_{q^h}$ over ${\mathbb F}_{q}$.

We will take a geometrical approach by associating a geometric object to an additive code. 

 Let $C$ be an additive $(n,q^{kh},d)_{q^h}$ code which is linear over ${\mathbb F}_q$ with a generator matrix $\mathrm{G}$. A column of $\mathrm{G}$ is a vector $v$ of ${\mathbb F}_{q^h}^{hk}$, for which we can write 
$$
v=\sum_{i=1}^{h} \epsilon_i v_i,
$$
for some $v_i \in {\mathbb F}_q^{hk}$. We can view the subspace spanned (over ${\mathbb F}_q$) by 
$$\{ v_i \ | \ i\in \{1,\ldots,h\}\}$$
as a subspace of $\mathrm{PG}(kh-1,q)$ which has (projective) dimension at most $h-1$.

Let $\mathcal X$ be the multi-set of these subspaces, so $\mathcal X$ is a multi-set of $n$ subspaces of $\mathrm{PG}(kh-1,q)$ each element of which is a subspace of dimension at most $h-1$.

Vice-versa, given a multi-set of subspaces of of $\mathrm{PG}(kh-1,q)$ of dimension at most $h-1$, after fixing a basis for the space, we can construct an additive code $C$ by reversing the above process. 

\begin{theorem} \label{genarcMDS}
Let $C$ be an additive $(n,q^{kh},d)_{q^h}$ code which is linear over ${\mathbb F}_q$ and let $\mathcal X$ be the multi-set of subspaces of $\mathrm{PG}(kh-1,q)$ obtained from a generator matrix $\mathrm{G}$ for $C$, as described above. A hyperplane of $\mathrm{PG}(kh-1,q)$ contains at most $n-d$ elements of $\mathcal X$ and some hyperplane contains exactly $n-d$ elements of $\mathcal X$. Vice-versa, if we have such a set of subspaces then we can reverse the process and obtain an additive $(n,q^{kh},d)_{q^h}$ code which is linear over ${\mathbb F}_q$.
\end{theorem}

\begin{proof}
For a non-zero vector $a \in {\mathbb F}_{q}^{hk}$, let $\pi_a$ be the hyperplane of $\mathrm{PG}(kh-1,q)$ corresponding to the hyperplane of ${\mathbb F}_{q}^{hk}$ orthogonal to the vector $a$. Let $x \in \mathcal X$ and let $v$ be the corresponding column in $\mathrm{G}$. The hyperplane $\pi_a$ contains $x$ if and only if $v \cdot a=0$. Since $a \cdot \mathrm{G}$ is a codeword it has, by Lemma~\ref{minweight}, at most $n-d$ zeros, which implies that  $v \cdot a=0$ for at most $n-d$ columns of $\mathrm{G}$, which in turn implies that at most $n-d$ subspaces of $\mathcal X$ are contained in $\pi_a$. Moreover, since the minimum distance of $C$ is $d$, there is some hyperplane $\pi_a$ which contains exactly $n-d$ elements of $\mathcal X$. Clearly, we can also reverse the argument, constructing an additive $(n,q^{kh},d)_{q^h}$ code from such a set of subspaces.
\end{proof}

Two codes $C$ and $C'$ over an alphabet $A$ are {\em equivalent} if one can be obtained from the other by a permutation of the coordinates and by permutations of the elements of $A$ in any coordinate.

\begin{theorem} \label{addcodeequiv}
Let $C$ and $C'$ be additive $(n,q^{kh},d)_{q^h}$ codes which are linear over ${\mathbb F}_q$. Let $\mathcal X$ (resp. $\mathcal X'$) be the multi-set of subspaces obtained from a generator matrix $\mathrm{G}$ (resp. $\mathrm{G}'$) for $C$ (resp. $C'$). If there is an element $\sigma \in \mathrm{P}\Gamma\mathrm{L}(kh,q)$ such that $\sigma(\mathcal X)=\mathcal X'$ then $C$ and $C'$ are equivalent.
\end{theorem}

\begin{proof}
An element $\sigma \in \mathrm{P}\Gamma\mathrm{L}(kh,q)$ acts on $\mathcal X$, by applying a field automorphism to the elements of $\mathcal X$ and then by left multiplying $\mathrm{G}$ by a non-singular $kh \times kh$ matrix. Applying a field automorphism to the matrix $\mathrm{G}$ simply permutes the elements of ${\mathbb F}_q$ in each coordinate. Multiplying $\mathrm{G}$ by a non-singular $kh \times kh$ matrix, simply replaces $G$ by another generator matrix for $C$. Ordering the elements of $\mathcal X$ is equivalent to ordering the coordinates of the elements of $C$. 
\end{proof}

Let $\mathcal X$ be a multi-set of subspaces of dimension at most $h-1$ of $\mathrm{PG}(kh-1,q)$. As mentioned before Theorem~\ref{genarcMDS}, if we fix a basis for this space we can obtain an additive code $C$ over ${\mathbb F}_{q^h}$, which is linear over ${\mathbb F}_q$. If we change the basis for $\mathrm{PG}(kh-1,q)$ then we will obtain a code equivalent to $C$, so the code does not depend on which basis we choose. Now suppose that we have two sets of  subspaces $\mathcal X$ and $\mathcal X'$ from which we construct additive codes $C$ and $C'$ respectively.
We say the codes $C$ and $C'$ are $\mathrm{P}\Gamma\mathrm{L}$-equivalent if there exists an element $\sigma \in \mathrm{P}\Gamma\mathrm{L}(kh,q)$ such that $\sigma(\mathcal X)=\mathcal X'$. 
It is a seemingly overlooked question as to whether the converse statement in Theorem~\ref{addcodeequiv} is true. Explicitly, if two additive codes $C$ and $C'$ are equivalent then does there exist an element $\sigma \in \mathrm{P}\Gamma\mathrm{L}(kh,q)$ such that $\sigma(\mathcal X)=\mathcal X'$? Thus, we are asking that if two additive codes $C$ and $C'$ are equivalent, then are they necessarily $\mathrm{P}\Gamma\mathrm{L}$-equivalent? Equivalently, if two additive codes $C$ and $C'$ are $\mathrm{P}\Gamma\mathrm{L}$-inequivalent then is it true that they are inequivalent? One can ask the same question for linear codes.

An isometry map from a code $C$ to a code $C'$ is a Hamming distance preserving bijection. We observe that $\mathrm{P}\Gamma\mathrm{L}$-equivalence is equivalent to semi-linear isometric equivalence defined in \cite{BBFKKW2006}. MacWilliams proved in her thesis \cite{MacWilliams1962} that there exists an element $\sigma \in \mathrm{PG}\mathrm{L}(kh,q)$ such that $\sigma(\mathcal X)=\mathcal X'$ if and only if there is a linear isometry between linear codes $C$ and $C'$. This is in the same spirit as the above question, but it is not the same question. For a proof of this MacWilliams theorem, see \cite{BGG1978} or \cite{WW1996}.

The {\em dual code} of $C$ is defined as
$$
C^{\perp}=\{ u \in {\mathbb F}_{q^h}^{n} \ | \ \mathrm{tr}_{q^h\rightarrow q} (u \cdot v)=0, \ \mathrm{for} \ \mathrm{all} \ v \in C\}.
$$

\begin{theorem} \label{dualadd}
The dual of an additive $(n,q^{kh},d)_{q^h}$ code $C$ which is linear over ${\mathbb F}_q$ is an additive $(n,q^{(n-k)h},d')_{q^h}$ code $C^{\perp}$ which is linear over ${\mathbb F}_q$.
\end{theorem}

\begin{proof}
Since $C^{\perp}$ is an orthogonal subspace to $C$, considering it as a subspace over ${\mathbb F}_q$, we have that $|C^{\perp}|=q^{(n-k)h}$. Thus, $C^{\perp}$ is an additive $(n,q^{(n-k)h},d')_{q^h}$ code, which is linear over ${\mathbb F}_q$.
\end{proof}

\section{Additive MDS codes over a finite field}

There are many well known results regarding linear codes over a finite field which carry through to additive MDS codes over a finite field. In this section, we prove some of these results.

We define $\mathcal X$ as an {\em arc of $(h-1)$-spaces} if it is a set of $(h-1)$-dimensional spaces of $\mathrm{PG}(kh-1,q)$ with the property that any $k$ elements of $\mathcal X$ span the entire space.

An important remark here is that an arc $\mathcal X$ of $(h-1)$-spaces of $\mathrm{PG}(2h-1,q)$ is a partial spread. In other words, the condition that any two elements of $\mathcal X$ span the whole space, is that they are skew.

\begin{theorem} \label{arcMDS}
Let $C$ be an additive $(n,q^{kh},d)_{q^h}$ code which is linear over ${\mathbb F}_q$ and let $\mathcal X$ be the set of subspaces obtained from a generator matrix $\mathrm{G}$ for $C$. Then $\mathcal X$ is an arc of $(h-1)$-spaces if and only if $C$ is an MDS code.
\end{theorem}

\begin{proof}
Let $\mathcal X$ be a set of subspaces obtained from a generator matrix $\mathrm{G}$ for $C$. Any $k$ subspaces of $\mathcal X$ span $\mathrm{PG}(kh-1,q)$ iff every hyperplane contains at most $k-1$ subspaces of $\mathcal X$. By Theorem~\ref{genarcMDS}, a hyperplane contains at most $k-1$ subspaces of $\mathcal X$ if and only if $C$ is an MDS code.
\end{proof}

In the construction of the set $\mathcal X$, we wrote each element of ${\mathbb F}_{q^h}$ in $\mathrm{G}$ with respect to a basis for ${\mathbb F}_{q^h}$ over ${\mathbb F}_q$. In this way we consider a generator matrix $\mathrm{G}$ for $C$ as $kh \times nh$ matrix with entries from ${\mathbb F}_{q}$. If we arbitrarily split the rows into disjoint sets of $h$ rows, whilst maintaining the natural partition of the columns we can also consider $\mathrm{G}$ as a $k\times n$ matrix whose entries are $h\times h$ matrices with entries from ${\mathbb F}_q$. The MDS property is given by the following theorem.

\begin{theorem} \label{hbyh}
A $k\times n$ matrix whose entries are $h\times h$ matrices with entries from ${\mathbb F}_q$ is a generator matrix for an additive $(n,q^{kh},n-k+1)_{q^h}$ MDS code if and only if every $k \times k$ submatrix of $\mathrm{G}$, considered as a $kh \times kh$ matrix, is non-singular.
\end{theorem}

\begin{proof}
This follows from Theorem~\ref{arcMDS}. 
\end{proof}

Theorem~\ref{hbyh} can be compared to similar statements for MDS codes over rings, for example, see \cite[Theorem 3]{Shiromoto2000}.

\begin{theorem} \label{dualMDS}
The dual of an additive $(n,q^{kh},n-k+1)_{q^h}$ MDS code $C$ which is linear over ${\mathbb F}_q$ is an additive $(n,q^{(n-k)h},k+1)_{q^h}$ MDS code $C^\perp$ which is linear over ${\mathbb F}_q$.
\end{theorem}

\begin{proof}
By Theorem~\ref{dualadd} and Lemma~\ref{minweight}, we only have to prove that the minimum weight of $C^{\perp}$, as a code over ${\mathbb F}_{q^h}$, is $k+1$. Suppose that $C^{\perp}$ has an element $v$ of weight at most $k$. Consider the set of $k$ vectors of ${\mathbb F}_{q^h}^{kh}$ obtained from the columns of a generator matrix for $C$ corresponding to the non-zero coordinates of $v$. Writing these vectors out over ${\mathbb F}_q$ we obtain a set of $kh$ linearly dependent vectors of ${\mathbb F}_{q}^{kh}$. This implies that the corresponding $k$ subspaces of $\mathcal X$ do not span the entire space, contradicting Theorem~\ref{arcMDS}. 
\end{proof}

The following lemma is a generalisation of what is sometimes called the {\em projection lemma} for arcs.

\begin{lemma}\label{projlemma} 
Let $\mathcal X$ be an arc of $(h-1)$-spaces in $\mathrm{PG}(kh-1,q)$ and let
$\mathcal A$ be a subset of $\mathcal X$ of size $r \leqslant k-2$. The projection of the elements of $\mathcal X \setminus \mathcal A$ from the subspace spanned by the elements of $\mathcal A$ is an arc of $(h-1)$-spaces in $\mathrm{PG}((k-r)h-1,q)$.
\end{lemma}

\begin{proof}
If a subset of $k-r$ subspaces of $\mathcal X \setminus \mathcal A$ do not span $\mathrm{PG}((k-r)h-1,q)$ then, together with the elements of $\mathcal A$, they do not span $\mathrm{PG}(kh-1,q)$, contradicting the arc property.
\end{proof}

\section{Projections onto Desarguesian spreads}

The vector space ${\mathbb F}_{q^h}$ is isomorphic to ${\mathbb F}_{q}^h$ when viewed as a vector space over ${\mathbb F}_q$. Under this isomorphism, we get a map $\Psi$ from the $(r-1)$-dimensional subspaces of $\mathrm{PG}(k-1,q^h)$ to the $(rh-1)$-dimensional subspaces of $\mathrm{PG}(kh-1,q)$. This is called the field reduction map in \cite{LVdV2015}, to which we refer to for more details. The image of the points of $\mathrm{PG}(k-1,q^h)$ under $\Psi$ is a {\em Desarguesian spread} of $(h-1)$-dimensional subspaces of $\mathrm{PG}(kh-1,q)$. If $\mathcal X$ is an arc of points of $\mathrm{PG}(k-1,q^h)$ then $\Psi(\mathcal X)$ is an arc of $(h-1)$-spaces of $\mathrm{PG}(kh-1,q)$.

\begin{lemma}
If $\mathcal X$ is arc of $(h-1)$-dimensional subspaces of $\mathrm{PG}(kh-1,q)$ which is contained in a Desarguesian spread then $\Psi^{-1}(\mathcal X)$ is an arc of $\mathrm{PG}(k-1,q^h)$.
\end{lemma}

To be able to identify an arc of lines of $\mathrm{PG}(5,q)$ which is contained in a Desarguesian spread we will use Theorem~\ref{Desproj}. We introduce some terminology which we will need in the proof of Theorem~\ref{Desproj}. 

A {\em partial spread set} $S$ is a set of $h \times h$ matrices with entries from ${\mathbb F}_q$, with the property that for all $\mathrm{A},\mathrm{B} \in S$, $\mathrm{A}\neq \mathrm{B}$, the matrix $\mathrm{A}-\mathrm{B}$ is non-singular, i.e. $\det(\mathrm{A}-\mathrm{B}) \neq 0$.

For each $\mathrm{A} \in S$, we define an $(h-1)$-dimensional subspace $\pi_{\mathrm{A}}$, to be the subspace spanned by the columns of the matrix
$$
\left(\begin{array}{c} \mathrm{I}_h \\ \hline \mathrm{A} \end{array} \right),
$$
where $\mathrm{I}_h$ is the $h \times h$ identity matrix.

Let $\pi_{\infty}$ be the $(h-1)$-dimensional subspace spanned by the columns of the matrix
$$
\left(\begin{array}{c} \mathrm{O}_h \\ \hline \mathrm{I}_h \end{array} \right),
$$
where $\mathrm{O}_h$ is the $h \times h$ zero matrix.

The name ``partial spread set'' derives from the fact that
$$
\{ \pi_{\mathrm{A}} \ | \ \mathrm{A} \in S \} \cup \{ \pi_{\infty} \}
$$
is a partial spread, i.e. an arc of $(h-1)$-dimensional subspaces of $\mathrm{PG}(2h-1,q)$.

We will need a converse of this statement, which we prove in the following lemma.

\begin{lemma} \label{spreadsetlemma}
A partial spread (i.e. an arc) $\mathcal X$ of $(h-1)$-dimensional subspaces of $\mathrm{PG}(2h-1,q)$ which contains $\pi_{\infty}$ is given by a partial spread set. 

If $h=2$ then we can assume that the partial spread set is
$$
S=\left\{ \left(\begin{array}{cc} x_1 & f_1(x_1,x_2) \\ x_2 & f_2(x_1,x_2) \end{array} \right) \ | \ (x_1,x_2) \in T \right\},
$$
for some subset $T$ of ${\mathbb F}_q^2$. Moreover, if the partial spread is contained in a Desarguesian spread then we can take $f_1$ and $f_2$ to be linear and if $\mathrm{O}_2 , \mathrm{I}_2 \in S$ then
$$
f_1(x_1,x_2)=a_1x_2, \ \mathrm{and} \ f_2(x_1,x_2)=x_1+a_2x_2,
$$
for some $a_1,a_2 \in {\mathbb F}_q$ such that $X^2+a_2X-a_1$ is an irreducible polynomial in ${\mathbb F}_q[X]$.
\end{lemma}

\begin{proof}
Let $\pi_{\mathrm{A}}$ be an element of the arc $\mathcal X$. Since $\pi_{\infty}$ is contained in $\pi$, the $h$-dimensional subspace $X_2=\cdots=X_{h}=0$, there is a unique point in the intersection of $\pi_{\mathrm{A}}$ and $\pi$, which is necessarily of the form 
$$
(1:0:\cdots:0:x_1:\cdots:x_h),
$$
where we use the separator $:$ to indicate that this is a projective point. The same argument works for any of the first $h$ coordinates, so we conclude that each element of the arc is the span of the columns of a matrix of the form 
$$
\left(\begin{array}{c} \mathrm{I}_h \\ \hline \mathrm{A} \end{array} \right).
$$
Moreover, since $\mathcal X$ is an arc, the set of matrices 
$$
\{ \mathrm{A} \ | \ \pi_{\mathrm{A}} \in \mathcal X \}
$$
is a partial spread set.

Clearly, any two elements of a partial spread set differ in the $i$-th column, for any $i \in \{1,\ldots,h\}$, so in particular we can assume that if $h=2$ the partial spread set has the desired form.

If the partial spread is contained in a Desarguesian spread then we can take $f_1$ and $f_2$ to be linear. This follows from \cite[pp. 220]{Dembowski1997}. Moreover, if $\mathrm{O}_2, \mathrm{I}_2 \in S$ then $f_1(0,0)=f_2(0,0)=0$, $f_1(1,0)=0$ and $f_2(1,0)=1$, from which it follows that 
$$
f_1(x_1,x_2)=a_1x_2, \ \mathrm{and} \ f_2(x_1,x_2)=x_1+a_2x_2,
$$
for some $a_1,a_2 \in {\mathbb F}_q$. 

The fact that 
$$
\left( \begin{array}{cc} x_1 & a_1x_2 \\ x_2 & x_1+a_2x_2 \end{array} \right) \ \mathrm{and} \
\left( \begin{array}{cc} 0 & 0 \\ 0 & 0 \end{array} \right)
$$
are in the spread set of a Desarguesian spread containing $\mathrm{O}_2$ and $\mathrm{I}_2$ implies that
$$
\left| \begin{array}{cc} x_1 & a_1x_2 \\ x_2 & x_1+a_2x_2 \end{array} \right| \neq 0,
$$
which implies that $X^2+a_2X-a_1$ is an irreducible polynomial in ${\mathbb F}_q[X]$.
\end{proof}

We are now in a position to prove the main theorem of the article.

\begin{theorem} \label{Desproj}
Let ${\mathcal X}$ be an arc of at least $q+k$ lines of $\mathrm{PG}(2k-1,q)$. If there is a subset ${\mathcal S}$ of $\mathcal X$ of size $k+1$ with the property that the projection of $\mathcal X$ from any $(k-2)$-subset of ${\mathcal S}$ is contained in a Desarguesian spread of $\mathrm{PG}(3,q)$ then ${\mathcal X}$ is contained in a Desarguesian spread.
\end{theorem}

\begin{proof}
We will prove the statement first for $k=3$.

After choosing a suitable basis, we can suppose that $\mathcal X$ is a set of lines whose $i$-th line $ \ell_i$ is spanned by the $(2i-1)$-th and $2i$-th column of a matrix of the form
$$
\left(
\begin{array}{cc|cc|cc|cc|c}
1 & 0 & 0 & 0 & 0 & 0 & 1 & 0 & \ldots \\
0 & 1 & 0 & 0 & 0 & 0 & 0 & 1 &  \ldots \\
0 & 0 & 1 & 0 & 0 & 0 & 1 & 0 & \ldots \\
0 & 0 & 0 & 1 & 0 & 0 & 0 & 1 & \ldots \\
0 & 0 & 0 & 0 & 1 & 0 & 1 & 0 &  \ldots \\
0 & 0 & 0 & 0 & 0 & 1 & 0 & 1 &  \ldots \\
\end{array}
\right),
$$
and that the projection from $\ell_i$, $i \in \{1,2,3,4\}$, is contained in a Desarguesian spread.

By hypothesis, the projection from $\ell_2$ and $\ell_3$ are partial spreads contained in Desarguesian spreads whose partial spread sets contain the zero matrix and the identity matrix. Hence, by Lemma~\ref{spreadsetlemma}, the above matrix is of the form
$$
\left(
\begin{array}{cc|cc|cc|cc|cc|c}
1 & 0 & 0 & 0 & 0 & 0 & 1 & 0 & 1 & 0 & \ldots \\
0 & 1 & 0 & 0 & 0 & 0 & 0 & 1 & 0 & 1 &  \ldots \\
0 & 0 & 1 & 0 & 0 & 0 & 1 & 0 & x_1 & a_1x_2 & \ldots \\
0 & 0 & 0 & 1 & 0 & 0 & 0 & 1 & x_2 & x_1+a_2x_2 & \ldots \\
0 & 0 & 0 & 0 & 1 & 0 & 1 & 0 & y_1 & b_1y_2 &  \ldots \\
0 & 0 & 0 & 0 & 0 & 1 & 0 & 1 & y_2 & y_1+b_2y_2 &  \ldots \\
\end{array}
\right),
$$
where the columns are given by $(x_1,x_2) \in T$, for some subset $T \subseteq {\mathbb F}_q^2$ of size $|\mathcal X|-4$. Observe that $y_1$ and $y_2$ are functions of $x_1$ and $x_2$, so $y_i=y_i(x_1,x_2)$ for $i \in \{1,2\}$.

To prove $\mathcal X$ is contained in a Desarguesian spread, we want to show that $a_1=b_1$ and $a_2=b_2$, i.e. the projections from $\ell_2$ and $\ell_3$ are onto {\em coherent} Desarguesian spreads. In fact in some case we shall prove, equivalently, that the projections from $\ell_1$ and $\ell_4$ are onto coherent Desarguesian spreads.

By hypothesis, the projection from $\ell_1$ is contained in a Desarguesian spread. Observe that $\ell_4$ is projected onto a line corresponding to the identity matrix in the partial spread set, $\ell_2$ is projected onto a line corresponding to the zero matrix and $\ell_3$ is projected onto $\pi_{\infty}$. Hence, by Lemma~\ref{spreadsetlemma},
$$
\left\{ \left( \begin{array}{cc}  y_1 & b_1y_2  \\y_2 & y_1+b_2y_2 \end{array} \right)\left( \begin{array}{cc}  x_1 & a_1x_2  \\x_2 & x_1+a_2x_2 \end{array} \right)^{-1} \ | \ (x_1,x_2) \in T \right\}
$$
is equal to  
$$
\left\{ \left( \begin{array}{cc}  z_1 & c_1z_2  \\z_2 & z_1+c_2z_2 \end{array} \right) \ | \ (x_1,x_2) \in T \right\}
$$
where $z_i=z_i(x_1,x_2)$, for $i \in \{1,2\}$. 

Therefore,
$$
\begin{array}{rcl}
y_1 & = & x_1z_1+c_1x_2z_2, \\ 
y_2 & = & x_1z_2+x_2(z_1+c_2z_2) \\
b_1y_2 & = &a_1x_2z_1+(x_1+a_2x_2)c_1z_2,\\
 y_1+b_2y_2 & = & a_1x_2z_2+(x_1+a_2x_2)(z_1+c_2z_2).
\end{array}
$$

If $z_2=0$ then the projection of $\mathcal X \setminus \{\ell_1,\ell_2\}$ from $\ell_1$ has at most $q$ lines, so $\mathcal X$ has size at most $q+2$, a contradiction. Hence, we have that $z_2 \neq 0$.

Thus, eliminating $y_1,y_2$ and $z_1$, these equations imply
$$
((a_1-b_1)(c_2-b_2)+(a_2-b_2)(b_1-c_1))x_1=((c_1a_2-c_2a_1)(b_2-a_2)+(c_1-a_1)(b_1-a_1))x_2.
$$

If not both of the coefficients of $x_1$ and $x_2$ are zero then $x_1$ and $x_2$ satisfy a linear relation which implies $\mathcal X \setminus \{\ell_2,\ell_3\}$  has at most $q$ lines. This implies that $\mathcal X$ has at most $q+2$ lines, again a contradiction.

Thus, both are zero and we can solve for $c_1$ and $c_2$ in terms of $a_1,a_2,b_1,b_2$. Explicitly the system is 
$$
\begin{array}{rcl}
(b_2-a_2)c_1+(a_1-b_1)c_2 & = & b_2a_1-a_2b_1,\\
((b_1-a_1)+a_2(b_2-a_2)c_1+a_1(a_2-b_2)c_2 & = &a_1(b_1-a_1).
\end{array}
$$
If $a_1=b_1$ and $a_2=b_2$ then the projections from $\ell_2$ and $\ell_3$ are onto coherent Desarguesian spreads, which is what we want to prove. If not then $c_1$ and $c_2$ have a unique solution, since otherwise $(a_1-b_1)/(a_2-b_2)$ is a root of $X^2+a_2X-a_1$, which, as we already observed in Lemma~\ref{spreadsetlemma}, is an irreducible polynomial in ${\mathbb F}_q[X]$.

By hypothesis, the projection from $\ell_4$ is also contained in a Desarguesian spread.  To obtain the projection from $\ell_4$, we project the elements of $\mathcal X$ onto $X_1=X_2=0$. Observe that $\ell_1$ is projected onto a line corresponding to the identity matrix in the partial spread set, $\ell_2$ is projected onto a line corresponding to the zero matrix, and $\ell_3$ is projected onto $\pi_{\infty}$. The line of $\mathcal X$ parameterised by $(x_1,x_2) \in T$ is projected onto the line spanned by the columns of 
$$
\left(
\begin{array}{cc}
 0 & 0  \\
0 & 0 \\
1- x_1 & -a_1x_2  \\
- x_2 & 1-x_1-a_2x_2 \\
1- y_1 &- b_1y_2 \\
- y_2 & 1-y_1-b_2y_2  \\
\end{array}
\right).
$$
Now, by Lemma~\ref{spreadsetlemma}, we have that
$$
\left\{ \left( \begin{array}{cc}  1-y_1 & -b_1y_2  \\-y_2 & 1-y_1-b_2y_2 \end{array} \right)\left( \begin{array}{cc}  1-x_1 & -a_1x_2  \\ -x_2 & 1-x_1-a_2x_2 \end{array} \right)^{-1} \ | \ (x_1,x_2) \in T \right\}
$$
is equal to  
$$
\left\{ \left( \begin{array}{cc}  w_1 & d_1w_2  \\w_2 & w_1+d_2w_2 \end{array} \right) \ | \ (x_1,x_2) \in T \right\}
$$
where $w_i=w_i(x_1,x_2)$, for $i \in \{1,2\}$.
Employing exactly the same argument as in the projection from $\ell_1$, 
$$
((a_1-b_1)(d_2-b_2)+(a_2-b_2)(b_1-d_1))(1-x_1)=((d_1a_2-d_2a_1)(b_2-a_2)+(d_1-a_1)(b_1-a_1))(-x_2).
$$
If at least one of the coefficients of $1-x_1$ and $-x_2$ is non-zero then $x_1$ and $x_2$ satisfy an affine linear relation which implies $\mathcal X \setminus \{\ell_2,\ell_3\}$  has at most $q$ lines. This implies that $|\mathcal X| \leqslant q+2$.

Thus, both are zero and we can solve for $d_1$ and $d_2$ in terms of $a_1,a_2,b_1,b_2$ and conclude that $c_1=d_1$ and $c_2=d_2$.

Hence, the projections from $\ell_1$ and $\ell_4$ are coherent from which it follows that $\mathcal X$ is contained in a Desarguesian spread.

We now prove the statement for $k\geqslant 4$. As in the case $k=3$ we can suppose, after a suitable change of basis, that $\mathcal X$ is a set of lines whose $i$-th line $ \ell_i$ is spanned by the $(2i-1)$-th and $2i$-th column of a matrix of the form
$$
\left(
\begin{array}{cc|cc|cc|cc|cc|cc|c}
1 & 0 & 0 & 0 &  \ldots & \ldots & 0 & 0 &1 & 0 & 1 & 0 & \ldots \\
0 & 1 & 0 & 0 &  \ldots & \ldots & 0 & 0 & 0 & 1 & 0 & 1 &  \ldots \\
0 & 0 & 1 & 0 &  \ldots &  \ldots & 0 & 0 & 1 & 0 & x_{11} & a_{11}x_{12} & \ldots \\
0 & 0 & 0 & 1 &  \ddots &  \ldots & 0 & 0 &0 & 1 & x_{12} & x_1+a_{12}x_{12} & \ldots \\
\vdots & \vdots &  \ldots &  \ddots & \ddots & \ddots & \ddots & \vdots & \vdots & \vdots & \vdots& \vdots \\
\vdots & \vdots &  \ldots &  \ldots & \ldots & \ddots & \ddots & \vdots  & \vdots & \vdots & \vdots & \vdots\\
0 & 0 & 0 & 0 &  \ldots &  \ldots & 1 & 0 & 1 & 0 & x_{k-1,1} & a_{k-1,1}x_{k-1,2} &  \ldots \\
0 & 0 & 0 & 0 & \ldots &  \ldots &  0 & 1 & 0 & 1 &x_{k-1,2} & x_{k-1,1}+a_{k-1,2}x_{k-1,2} &  \ldots \\
\end{array}
\right),
$$
where we suppose that the projection hypothesis is satisfied by the lines in the set $L=\{\ell_1,\ldots,\ell_{k+1}\}$. Let $M$ be any $(k-3)$-subset of  $L\setminus \{\ell_1,\ell_{k+1}\}$. The projection onto an arc of lines of $\mathrm{PG}(5,q)$ from $M$, satisfies that hypothesis of the theorem for the lines of $L \setminus M$ for $k=3$. Therefore, we have that all the spread sets are coherent and that the arc is contained in a Desarguesian spread.  
\end{proof}

We now reformulate Theorem~\ref{Desproj} for codes, slightly weakening the hypothesis to make it more readable. By Lemma~\ref{projlemma}, projecting an arc $\mathcal X$ of $(h-1)$-dimensional spaces of $\mathrm{PG}(kh-1,q)$, from a fixed $(h-1)$-dimensional subspace, is an arc of $(h-1)$-dimensional spaces of $\mathrm{PG}((k-1)h-1,q)$. Therefore, an additive $(n,q^{kh},n-k+1)_{q^h}$ additive MDS code projects onto an $(n,q^{(k-1)h},n-k+1)_{q^h}$ additive MDS code. Some authors call this process {\em shortening}. This should not be confused with {\em truncating} a code which is simply deleting some coordinates from all the codewords, which in our case is equivalent to taking a subset of $\mathcal X$. Blokhuis and Brouwer \cite[Proposition 3.1]{BB2004} proved that if all the truncations of an additive code over ${\mathbb F}_4$ are even-dimensional then the code is in fact linear over ${\mathbb F}_4$. That statement can be compared to the following theorem, which follows directly from Theorem~\ref{Desproj}. 

\begin{theorem} \label{projcodesthm}
If $n \geqslant q+k \geqslant q+3$ and $k+1$ of the projections of an $(n,q^{2k},n-k+1)_{q^2}$ additive MDS code $C$, which is linear over ${\mathbb F}_q$, are linear over ${\mathbb F}_{q^2}$ then $C$ is linear over ${\mathbb F}_{q^2}$. 
\end{theorem}

\begin{proof}
Let $\mathcal X$ be the arc of $(h-1)$-dimensional subspaces obtained from the columns of a generator matrix for $C$. By hypothesis, there is a subset $\mathcal S$ of $\mathcal X$ of size $k+1$, such that the projection of $\mathcal X$ from an element of $\mathcal S$ is contained in a Desarguesian spread. Therefore, if we project $\mathcal X$ from any $(k-2)$-subet of $S$ then this projection will be contained in a Desarguesian spread. By Theorem~\ref{Desproj}, this implies $\mathcal X$ is contained in a Desarguesian spread, which implies $C$ is linear over ${\mathbb F}_{q^2}$. 
\end{proof}

\section{Additive MDS codes over small finite fields}

The computational results were obtained in two ways. 

{\bf Method 1.} The first method used to classify arcs in $\mathrm{PG}(kh-1,q)$ started by classifying partial spreads in $\mathrm{PG}(2h-1,q)$ up to $\mathrm{P}\Gamma \mathrm{L}(2h,q)$ isomorphism, building on the work of Soicher \cite{Soicher}. These classifications were obtained using GAP \cite{GAP2020}, and in particular, the packages GRAPE \cite{GRAPE} and FinInG \cite{fining}. Of particular use was the
functionality in GRAPE to classify the maximal cliques of given size in a graph, 
up to the action of a given subgroup of the automorphism group of that
graph, as well as Linton's program {\tt SmallestImageSet} (included with GRAPE),
which determines the lexicographically least set in a $G$-orbit of sets, without
explicitly computing the $G$-orbit; see \cite{Linton2004}. Once a classification of arcs in $\mathrm{PG}((k-1)h-1,q)$ was obtained recursively, an exhaustive search was undertaken to determine those arcs of size $n$ which ``lifted" to arcs of size $n+1$ in $\mathrm{PG}(kh-1,q)$. These lifts were then checked for isomorphic copies.

{\bf Method 2.} The other method used a breadth first algorithm with isomorphism rejection under the collineation group of the ambient projective space at every step. This type of algorithm is described  in \cite[Section 9.6]{BBFKKW2006}, (where it is called {\it snakes and ladders}). The algorithms were implemented in GAP based on some of the functionality provided by FinInG, albeit with a lot of tweaking. The arcs are built up from the empty set adding one subspace at every step. For each size $n$, the algorithm stores a list of orbit representatives of arcs of size $n$, together with a Schreier vector (allowing constructive recognition) and the stabiliser of the representative. At the next step, for each representative, the algorithm computes the orbits on the set of all possible extensions (to an arc of size $n+1$) of the representative under the stabiliser of the representative.  This gives a list of all possible extensions of size $n+1$ of all representatives of arcs of size $n$. The list of arcs of size $n+1$ thus obtained is then reduced by removing the representatives from the list which belong to the same orbit as a representative which comes earlier in the list. This step uses the Schreier vectors which were stored at the previous steps. In the process the stabilisers of the new representatives are also computed.

\bigskip

\hrule

\bigskip

{Case $q^h=2^2$.}

\begin{center}
\begin{table}[h] 
\begin{tabular}{|r|c|c|c|} \hline
size & 4 & 5 & 6 \\ \hline
number of arcs of points of $\mathrm{PG}(2,4)$ & 1 &1 & 1 \\
number of arcs of lines of $\mathrm{PG}(5,2)$ & 1 &1 & 1 \\ \hline
\end{tabular}
\vspace{0.3cm}
\caption{The classification of arcs of lines of $\mathrm{PG}(5,2)$.}\label{codes4}
\end{table}
\end{center}
\vspace{-1.4cm}
The classification in Table~\ref{codes4} was obtained using Method 2. Suppose there exists a $(n,4^k,n-k+1)_4$ additive MDS code. By Theorem~\ref{dualMDS} and the trivial upper bound, we can assume that $k \leqslant \frac{1}{2}n \leqslant \frac{1}{2}(k+3)$ which gives $k \leqslant 3$. 
Thus, from {\sc Table} \ref{codes4}, we conclude that all additive MDS codes over ${\mathbb F}_4$ are equivalent to a linear code over ${\mathbb F}_4$. Indeed, it was already observed in \cite{KO2016}, that all MDS codes over an alphabet of size $4$ are equivalent to linear codes. This also follows from the Blokhuis-Brouwer theorem \cite[Proposition 3.1]{BB2004} mentioned before.

\bigskip

\hrule

\bigskip

{Case $q^h=2^3$.}

\begin{center}
\begin{table}[h] 
\begin{tabular}{|r|c|c|c|c|c|c|c|} \hline
size & 4 & 5 & 6 & 7 & 8 & 9 & 10\\ \hline
number of arcs of points of $\mathrm{PG}(2,8)$ & 1 & 1 & 3 & 2 & 2 & 2 & 1 \\
number of arcs of planes of $\mathrm{PG}(8,2)$ & 1 & 2 & 4 & 2 & 2 & 2 & 1 \\ \hline
\end{tabular}
\vspace{0.3cm}
\caption{The classification of arcs of planes of $\mathrm{PG}(8,2)$.}\label{codes8}
\end{table}
\end{center}
\vspace{-1.4cm}
The classification in {\sc Table}~\ref{codes8} was obtained using Method 2. {\sc Table}~\ref{codes8} implies that the longest $(n,8^3,n-2)_8$ additive MDS code which is not $\mathrm{P}\Gamma\mathrm{L}$-equivalent to a linear MDS code is a $(6,8^3,4)_8$ code. It is unique, up to isomorphism, and has the following generator matrix

$$
\left(
\begin{array}{ccc|ccc|ccc|ccc|ccc|ccc}
1 & 0 & 0 & 0 & 0 & 0 & 0 & 0 & 0 & 1 & 0 & 0 & 1 & 0 & 0 & 1 & 0 & 0\\
0 & 1 & 0 & 0 & 0 & 0 & 0 & 0 & 0 & 0 & 1 & 0 & 0 & 1 & 0 & 0 & 1 & 0\\
0 & 0 & 1 & 0 & 0 & 0 & 0 & 0 & 0 & 0 & 0 & 1 & 0 & 0 & 1 & 0 & 0 & 1\\
0 & 0 & 0 & 1 & 0 & 0 & 0 & 0 & 0 & 1 & 0 & 0 & 1  & 1 & 1 & 1 & 1 &  0  \\
0 & 0 & 0 & 0 & 1 & 0 & 0 & 0 & 0 & 0 & 1 & 0 & 1 & 0 & 0 &  0 & 1 &  1  \\
0 & 0 & 0 & 0 & 0 & 1 & 0 & 0 & 0 & 0 & 0 & 1 & 1 & 1 & 0 &  1 & 0 &  0  \\
0 & 0 & 0  & 0 & 0 & 0 & 1 & 0 & 0 & 1 & 0 & 0 & 0 & 1 & 0 &  0 & 1 &  1  \\
0 & 0 & 0 & 0 & 0 & 0 & 0 & 1 & 0 & 0 & 1 & 0 & 0 & 1 &  1 &  1 & 1 &  1 \\
0 & 0 & 0  & 0 & 0 & 0 & 0 & 0 & 1 & 0 & 0 & 1 & 1 & 0 & 1 &  0 & 1 &  0  \\
\end{array}
\right).
$$

In \cite{KO2016}, it was proven that all $(n,8^k,d)_8$ MDS codes with $d \geqslant 5$ are equivalent to linear codes and all $(n,8^k,4)_8$ MDS codes with $k \geqslant 4$ are equivalent to linear codes. There are precisely $39$ equivalence classes of $(6,8^3,4)_8$ codes. From the table above, assuming $\mathrm{P}\Gamma\mathrm{L}$-inequivalence implies inequivalence, we conclude that three of these contain a linear code and one contains an additive non-linear code. We conclude that the only additive MDS codes over ${\mathbb F}_8$, which are not $\mathrm{P}\Gamma\mathrm{L}$-equivalent to a linear code, are the $(5,8^3,3)_8$ (its dual which is a $(5,8^2,4)_8$ code coming from a partial spread of planes of $\mathrm{PG}(5,2)$) and the $(6,8^3,4)_8$ additive code, which must be necessarily equivalent to its dual.

\bigskip

\hrule

\bigskip

{Case  $q^h=3^2$.}

\begin{center}
\begin{table}[h] 
\begin{tabular}{|r|c|c|c|c|c|c|c|} \hline
size & 4 & 5 & 6 & 7 & 8 & 9 & 10\\ \hline
$\#$ of arcs of points of $\mathrm{PG}(2,9)$ & 1 & 2 & 6 & 3 & 2 & 1 & 1 \\
$\#$  of arcs of lines of $\mathrm{PG}(5,3)$ & 1 & 4 & 13 & 4 & 3 & 1 & 1 \\ \hline
\end{tabular}
\vspace{0.3cm}
\caption{The classification of arcs of lines of $\mathrm{PG}(5,3)$.} \label{codes9}
\end{table}
\end{center}
\vspace{-1.4cm}
The classification in {\sc Table}~\ref{codes9} was obtained using Method 1 for $n \geqslant 8$ and then checked and completed using Method 2. By {\sc Table}~\ref{codes9}, the longest $(n,9^3,n-2)_9$ additive MDS code which is not $\mathrm{P}\Gamma\mathrm{L}$-equivalent to a linear MDS code is a $(8,9^3,6)_9$ code. It is unique, up to $\mathrm{P}\Gamma\mathrm{L}$-equivalence, and has the following generator matrix

$$
\left(
\begin{array}{cc|cc|cc|cc|cc|cc|cc|cc}
1 & 0 & 0 & 0 & 0 & 0 & 1 & 0 & 1 & 0 & 1 & 0 & 1 & 0 & 1 & 0\\
0 & 1 & 0 & 0 & 0 & 0 & 0 & 1 & 0 & 1 & 0 & 1 & 0 & 1  & 0 & 1\\
0 & 0 & 1 & 0 & 0 & 0 & 1 & 0 & 2 & 2 & 0 & 2 & 2 & 0  & 1 & 2\\
0 & 0 & 0 & 1 & 0 & 0 & 0 & 1 & 0 & 2 & 1 & 1 & 1 & 2 & 1 & 0 \\
0 & 0 & 0 & 0 & 1 & 0 & 1 & 0 & 0 & 1 & 2 & 0 & 1 & 1 & 2 & 1\\
0 & 0 & 0 & 0 & 0 & 1 & 0 & 1 & 2 & 1 & 2 & 2 & 2 & 0 & 0 & 2\\
\end{array}
\right).
$$
As a code $C$ over ${\mathbb F}_9$, this implies that
$$
C=\{(x,y,z,x+y+z,x+e^3 y +y^3+e^3 z-e^2z^3,
x+e^3 y -y^3+e^3 z+e^2 z^3,$$
$$
x+e y -e^2 y^3+e z+ z^3,
x+e y +e^2 y^3+e z-z^3) \ | \ x,y,z \in {\mathbb F}_9 \},
$$
where $e$ is a primitive element of ${\mathbb F}_{9}$ satisfying $e^2=e+1$. It is possible to verify directly that three coordinates of the elements of $C$ are zero if and only if $x=y=z=0$, see \cite[Proposition 5.1]{Gamboa2020}. Furthermore, this code has the property that $C \leqslant C^{\perp_a}$, so Theorem~\ref{ketkarthm} implies that there is a $[\![8,2,4]\!]_3$ quantum MDS code, which does not come from a linear code over ${\mathbb F}_9$ contained in its Hermitian dual. There is an example of a $[\![8,2,4]\!]_3$ quantum MDS code, which does come from a linear code over ${\mathbb F}_9$ contained in its Hermitian dual, see \cite{GR2015}.

To determine the additive non-linear MDS codes over ${\mathbb F}_9$ for $k\geqslant 4$, by Theorem~\ref{dualMDS}, we can assume that $k\leqslant \frac{1}{2}n$. By Theorem~\ref{Desproj}, if we assume that the projection onto $k=3$ has size at most $6$, then we have a contradiction given by the inequalities
$$
4 \leqslant k \leqslant \tfrac{1}{2}n 
$$
and
$$
n-(k-3) \leqslant 6.
$$
Thus, to find additive non-linear MDS codes over ${\mathbb F}_9$ for $k\geqslant 4$, by Theorem~\ref{Desproj}, we can assume that at least one of the projections down to $k=3$ is onto one of the examples of either the $(7,9^3,5)_9$  or the $(8,9^3,6)_9$ additive non-linear MDS code. 

Method 1 was employed to reveal that there is no additive $(8,9^4,5)_9$ (resp. $(9,9^4,6)_9$) code which projects onto the $(7,9^3,5)_9$ (resp. $(8,9^3,6)_9$) code, so we conclude that the only additive non-linear MDS codes over ${\mathbb F}_9$, with $k \geqslant 3$ and $d \geqslant 4$, are the $(7,9^3,5)_9$ and the $(8,9^3,6)_9$ additive non-linear MDS code and their duals, which are $(7,9^4,4)_9$ and $(8,9^5,4)_9$ additive non-linear MDS codes.  

The $(n,9^2,n-1)_9$ additive non-linear MDS codes come from partial spreads of $\mathrm{PG}(3,3)$ and are tabulated in the following table. The ones that do not correspond to an arc of points in $\mathrm{PG}(1,9)$ correspond to additive non-linear codes. The duals of these codes are additive non-linear $(n,9^{n-2},3)_9$ MDS codes. These additive non-linear codes are not $\mathrm{P}\Gamma\mathrm{L}$-equivalent to a linear code.

\begin{center}
\begin{table}[h] 
\begin{tabular}{|r|c|c|c|c|c|c|c|} \hline
size & 4 & 5 & 6 & 7 & 8 & 9 & 10\\ \hline
$\#$ of arcs of points of $\mathrm{PG}(1,9)$ & 2 & 2 & 2 & 1 & 1 & 1 & 1 \\
$\#$  of arcs of lines of $\mathrm{PG}(3,3)$ & 3 & 4 & 5 & 4 & 3 & 2 & 2 \\ \hline
\end{tabular}
\vspace{0.3cm}
\caption{The classification of arcs of lines of $\mathrm{PG}(3,3)$.} \label{codes9b}
\end{table}
\end{center}

The classification in {\sc Table}~\ref{codes9b} was obtained using Method 2.

\bigskip

\hrule

\bigskip

{Case $q^h=4^2$.}

\begin{center}
\begin{table}[h] 
\begin{tabular}{|r|c|c|c|c|c|c|c|c|c|c|c|c|c|c|c|c|} \hline
size & 5 & 6 & 7 & 8 & 9 & 10 & 11 \\ \hline
$\#$ of arcs of $\mathrm{PG}(2,16)$ & 3 & 22 & 125 & 865 & 1534 & 1262 & 300\\
$\#$ of line-arcs of $\mathrm{PG}(5,4)$ & 10 & 360 & 8294 & 15162 & 2869 & 1465 & 301 \\ \hline
size & 12 & 13 & 14 & 15 & 16 & 17 & 18 \\ \hline
$\#$ of arcs of $\mathrm{PG}(2,16)$ &   159 & 70 & 30 & 9 & 5 & 3 & 2 \\
$\#$ of line-arcs of $\mathrm{PG}(5,4)$ &  159 & 70 & 30 & 9 & 5 & 3 & 2 \\\hline
\end{tabular}
\vspace{0.3cm}
\caption{The classification of arcs of lines of $\mathrm{PG}(5,4)$.} \label{codes16}
\end{table}
\end{center}
\vspace{-1.3cm}
The classification in {\sc Table}~\ref{codes16} was obtained using Method 1 for $n\geqslant 11$ employing Theorem~\ref{Desproj} and then checked and completed using Method 2. The table indicates that the longest $(n,16^3,n-2)_{16}$ additive MDS code which is not $\mathrm{P}\Gamma\mathrm{L}$-equivalent to a linear MDS code is an $(11,16^3,9)_{16}$ code. It is unique, up to $\mathrm{P}\Gamma\mathrm{L}$-equivalence, and has the following generator matrix

$$
\left(
\begin{array}{cc|cc|cc|cc|cc|cc|cc|cc|cc|cc|cc}
1 & 0 & 0 & 0 & 0 & 0 & 1 & 0 & 1 & 0 & 1 & 0 & 1 & 0 & 1 & 0 & 1 & 0 & 1 & 0 & 1 & 0\\
0 & 1 & 0 & 0 & 0 & 0 & 0 & 1 & 0 & 1 & 0 & 1 & 0 & 1 & 0 & 1 & 0 & 1 & 0 & 1 & 0 & 1 \\
0 & 0 & 1 & 0 & 0 & 0 & 1 & 0 & 0 & e & 1 & e^2 & 1 & 1 & 0 & 1 & e & e^2 & e & 1 & e & e \\
0 & 0 & 0 & 1 & 0 & 0 & 0 & 1 & e & e^2 & e & 0 & e^2 & 1 & 1 & e^2 & e^2 & e^2 & 0 & e & e & 1 \\
0 & 0 & 0 & 0 & 1 & 0 & 1 & 0 & 1 & e & e & e & e^2 & 0 & 0 & e & 0 & 1 & e & e^2 & e & 1 \\
0 & 0 & 0 & 0 & 0 & 1 & 0 & 1 & e & 0 & 0 & e & 1 & e & e & e^2 & 1 & e^2 & e & e & 1 & 1 \\
\end{array}
\right),
$$
where $e$ is a primitive element of ${\mathbb F}_{4}$ satisfying $e^2=e+1$.

As a code $C$ over ${\mathbb F}_{16}$ this implies that
$$
C=\{(x,y,z,x+y+z,x+\theta^{12} y +\theta ^4 y^4 +\theta z+\theta^{12}z^4,
x+\theta^2 y +\theta^7y^4+ z+\theta^7 z^4,
x+\theta^4 y^4 +\theta^2  z+ \theta^4z^4,$$
$$
x+\theta^{14}y+\theta^7 y+ \theta^{12}z+\theta^4z^4,
x+\theta y+\theta^9 y^4+\theta^{14} z+\theta^7 z^4,
x+\theta^2y^4+z+\theta^{14}z^4,$$
$$
x+\theta^{14}y+\theta^4y^4+\theta^{12}z+\theta^7z^4) 
\ | \ x,y,z \in {\mathbb F}_{16} \},
$$
where $\theta$ is a primitive element of ${\mathbb F}_{16}$ satisfying $\theta^4=\theta+1$. 

Method 1 reveals that the unique $(11,16^3,9)_{16}$ code does not lift to a $(12,16^4,9)_{16}$ code. We were not able to check the 203 additive non-linear  $(10,16^3,8)_{16}$ codes to see if they lift or not.

\begin{theorem}
The longest additive MDS codes over ${\mathbb F}_{16}$, which are linear over ${\mathbb F}_4$, are linear over ${\mathbb F}_{16}$.
\end{theorem}

\begin{proof}
Suppose there is a $(17,16^k,18-k)_{16}$ additive MDS code over ${\mathbb F}_{16}$, which is linear over ${\mathbb F}_4$. By Theorem~\ref{dualadd}, we can suppose $k \leqslant \frac{1}{2}n$, so $k \leqslant 8$. By Theorem~\ref{arcMDS}, the set of columns of a generator matrix for the code is an arc $\mathcal X$ of lines of $\mathrm{PG}(2k-1,4)$. By Lemma~\ref{projlemma}, this arc projects to an arc of lines of size $17-(k-3)=20-k \geqslant 12$ in $\mathrm{PG}(5,4)$, which by Table~\ref{codes16} is contained in a Desarguesian spread. 
By Theorem~\ref{Desproj}, this implies that $\mathcal X$ is contained in a Desarguesian spread, which implies that the code is linear.
\end{proof}

All $[17,k,18-k]_{16}$ linear MDS codes over ${\mathbb F}_{16}$ have been classified, see \cite{HS2001}. If $k \not\in \{ 3,14\}$ then a $[17,k,18-k]_{16}$ linear MDS code over ${\mathbb F}_{16}$ is a Reed-Solomon code. For $k=3$ there are $[18,3,16]_{16}$ linear MDS codes over ${\mathbb F}_{16}$ which can be obtained by extending the Reed-Solomon codes or by taking the code generated by the matrix whose columns are the points of a Lunelli-Sce hyperoval \cite{LS1964}. These codes shorten to $[17,3,15]_{16}$ linear MDS codes over ${\mathbb F}_{16}$ and dualise to $[18,15,4]_{16}$ linear MDS codes over ${\mathbb F}_{16}$ which truncate to $[17,15,3]_{16}$ MDS codes and project to $[17,14,4]_{16}$ MDS codes.

\section{The MDS conjecture and the quantum MDS conjecture}

The results in the previous section allow us to prove the MDS conjecture in the following cases. Note that the MDS conjecture is known to hold for all codes over ${\mathbb F}_4$ and ${\mathbb F}_8$, see \cite{KO2016}.

\begin{theorem}
The MDS conjecture holds for additive codes over ${\mathbb F}_9$ and for additive codes over ${\mathbb F}_{16}$ which are linear over ${\mathbb F}_4$.
\end{theorem}

\begin{proof}
In the case of $q=9$ we determined all additive non-linear codes and all are of length $n \leqslant 8$. Since the MDS conjecture has been verified for linear codes, see the survey article \cite{HS2001}, we conclude that the MDS conjecture holds.

For $q=16$, suppose there is a $(18,16^k,19-k)_{16}$ additive MDS code which is linear over ${\mathbb F}_4$. By Theorem~\ref{dualMDS}, we can assume that $k\leqslant 9$. By Theorem~\ref{arcMDS}, there is an arc of lines of $\mathrm{PG}(2k-1,4)$ of size $18$ which projects down to an arc of lines of $\mathrm{PG}(5,4)$ of size $21-k \geqslant 12$. In {\sc Table}~\ref{codes16}, we deduced that the largest arc of lines which is not contained in a Desarguesian spread is of size $11$. Hence, all projections are onto arcs of lines contained in a Desarguesian spread which, by Theorem~\ref{projcodesthm}, implies that the $(18,16^k,19-k)_{16}$ additive MDS code is linear over ${\mathbb F}_{16}$.
\end{proof}

The quantum MDS conjecture, initially mentioned in \cite{KKKS2006}, and again in \cite{HG2019}, states that apart from the case $d=4$ and $q=2^h$, if there is a $[\![n,n-2(d-1),d]\!]_q$ stabiliser MDS code with $d \geqslant 3$ then $n \leqslant q^2+1$. By Theorem~\ref{ketkarthm}, the existence of a $[\![q^2+2,q^2+2-2(d-1),d]\!]_q$ stabiliser MDS code implies the existence of a $(q^2+2,d-1,q^2+4-d)_{q^2}$ additive MDS code. The non-existence of such a code was already known for $q=2$ and here we have ruled out such codes for $q=3$. Thus, we conclude that the quantum MDS conjecture holds for $q \in \{2,3\}$ and if there is a counter-example for $q=4$ then it must come from an additive $(18,16^k,19-k)_{16}$ MDS code which is not linear over ${\mathbb F}_4$.

\section{Conclusions and further work}

One could try and extend Theorem~\ref{projcodesthm} to additive codes over ${\mathbb F}_{q^h}$, for $h \geqslant 3$.

There is a $(14,25^3,12)_{25}$ additive MDS code which is not $\mathrm{P}\Gamma\mathrm{L}$-equivalent to a linear MDS code, see \cite{Gamboa2020}. This corresponds to an arc of lines of $\mathrm{PG}(5,5)$ of size $14$. We have not been able to classify all partial spreads of lines (i.e. arcs of lines) in $\mathrm{PG}(3,5)$ but if this were feasible, it would most likely be feasible to classify the largest arcs of lines in $\mathrm{PG}(5,5)$. Likewise, a classification of partial spreads of planes of $\mathrm{PG}(5,3)$, would allow us to determine the largest arcs of planes in $\mathrm{PG}(8,3)$. If this largest examples turn out to be contained in a Desarguesian spread then we would verify that there are no additive non-linear codes over ${\mathbb F}_{25}$ (resp. ${\mathbb F}_{27}$) which outperform linear codes over ${\mathbb F}_{25}$ (resp. ${\mathbb F}_{27}$).

It is tempting to believe that relaxing the linearity constraint to additivity would allow the discovery of MDS codes which outperform their linear counterparts. However, it appears that this is not the case, as we have confirmed here for small fields. In fact, it appears that if we do not impose linearity then additive codes are not nearly as good as linear codes. This was also the conclusion in the work of \cite{Alderson2006}, \cite{KDO2015} and \cite{KO2016} for smaller alphabets, when one considers MDS codes with no constraints.

\section{Acknowledgments}

We thank Leonard Soicher for his help in classifying partial spreads in $\mathrm{PG}(3,4)$ and for useful comments on drafts of this paper.

\bigskip

{\small Simeon Ball and Guillermo Gamboa}  \\
{\small Departament de Matem\`atiques}, \\
{\small Universitat Polit\`ecnica de Catalunya, Jordi Girona 1-3},
{\small M\`odul C3, Campus Nord,}\\
{\small 08034 Barcelona, Spain} \\
{\small {\tt simeon.michael.ball@upc.edu, guillermo.gamboa@upc.edu}}

  \bigskip
  
Michel Lavrauw\\
Faculty of Engineering and Natural Sciences,\\
Sabanc\i \  University,\\
Istanbul, Turkey\\
{\tt mlavrauw@sabanciuniv.edu}

 


\begin{thebibliography}{}

\bibitem{Alderson2006} T. L. Alderson, $(6, 3)$-MDS codes over an alphabet of size $4$, {\it Des. Codes Cryptogr.}, {\bf 38} (2006) 31--40.

\bibitem{Ball2012}
S. Ball, On sets of vectors of a finite vector space in which every subset of basis size is a basis, {\it J. Eur. Math. Soc.}, {\bf 14} (2012) 733--748. 

\bibitem{BL2019}
S. Ball and M. Lavrauw, Arcs in finite projective spaces, {\it EMS Surv. Math. Sci.}, {\bf 6} (2019) 133--172.


\bibitem{BBFKKW2006} A. Betten, M. Braun, H. Fripertinger, A. Kerber, A. Kohnert, A. Wassermann, {\it Error-Correcting Linear Codes, Classification by Isometry and Applications}, Algorithms and Computation in Mathematics 18, Springer, 2006.


\bibitem{BB2004} A. Blokhuis and A. E. Brouwer, 
Small additive quaternary codes,
{\it European J. Combin.}, {\bf 25} (2004) 161--167. 

\bibitem{BGG1978} K. Bogart, D. Goldberg and J. Gordon, An elementary proof of the MacWilliams theorem on equivalence of codes, {\it Inform and Control}, {\bf 37} (1978) 19--22.

\bibitem{Bush1952}
K. A. Bush. Orthogonal arrays of index unity, {\it Ann. Math. Statistics}, {\bf 23} (1952) 426--434.

\bibitem{Dembowski1997} P. Dembowski, {\it Finite Geometries}, Springer, 1997 (Reprint of the original 1968 publication).

\bibitem{fining}
J.~Bamberg, A.~Betten, {Ph}. Cara, J.~De~Beule, M.~Lavrauw, and
  M.~Neunh\"offer.
\newblock {\em Finite Incidence Geometry}.
\newblock FinInG -- a {GAP} package, version 1.4.1, 2018.

\bibitem{Gamboa2020} G. A. Gamboa Quintero, Additive MDS codes, Master's thesis, Universitat Polit\`ecnica Catalunya, 2020.

\bibitem{GAP2020}
The GAP Group, GAP -- Groups, Algorithms, and Programming, Version 4.11.0, 2020. \url{https://www.gap-system.org}

\bibitem{GRAPE} L.H. Soicher, The GRAPE package for GAP, Version 4.8.3, 2019, \\ \url{https://gap-packages.github.io/grape}

\bibitem{GR2015} M. Grassl and M. R\"otteler, Quantum MDS codes over small fields, in {\it Proc. Int. Symp. Inf. Theory (ISIT)}, 1104--1108 (2015).

\bibitem{HS2001} J. W. P. Hirschfeld and L. Storme. The packing problem in statistics, coding theory and finite projective spaces: update 2001. Finite geometries, 201--246, Dev. Math., 3, Kluwer Acad. Publ., Dordrecht, 2001.

\bibitem{HG2019} F. Huber and M. Grassl, 
Quantum codes of maximal distance and highly entangled subspaces, {\it Quantum}, {\bf 4} 284 (2020).

\bibitem{KKKS2006} A. Ketkar, A. Klappenecker, S. Kumar, and P. K. Sarvepalli, Nonbinary stabilizer codes over finite fields, {\em IEEE Trans. Inf. Theory}, {\bf  52} (2006) 4892--4914.

\bibitem{KDO2015} J. I. Kokkala, D. S. Krotov and P. R. J. \"Osterg{\aa}rd, On the classification of MDS codes, {\it IEEE Trans. Inf. Theory}, {\bf 61} (2015) 6485--6492. 

\bibitem{KO2016} J. I. Kokkala and P. R. J. \"Osterg{\aa}rd,
Further results on the classification of MDS codes,
{\it Adv. Math. Commun.}, {\bf 10} (2016) 489--498. 

\bibitem{LVdV2015} M. Lavrauw and G, Van de Voorde, Field reduction and linear sets in finite geometry, in: {\it Contemporary Mathematics}, (eds: G Kyureghyan, GL Mullen, and A Pott), volume 632, pp. 271--293, American Mathematical Society, 2015.


\bibitem{Linton2004} S. Linton, Finding the smallest image of a set, in:
{\it ISSAC '04: Proceedings of the 2004 international symposium on symbolic and 
algebraic computation}, pp. 229--234, 2004. 


\bibitem{LS1964} L. Lunelli and M. Sce, Considerazione aritmetiche e risultati
sperimentali sui $\{K;n\}_q$-archi, {\em Ist. Lombardo Accad. Sci. Rend.
A}, {\bf 98} (1964) 3--52.

\bibitem{MacWilliams1962}
F. J. MacWilliams, {\it Combinatorial problems of elementary abelian groups}, Ph.D. Dissertation, Harvard University, Cambridge, Mass., 1962.

\bibitem{Shiromoto2000} K. Shiromoto, Note on MDS codes over the integers modulo $p^{m}$, {\it Hokkaido Mathematical Journal}, {\bf 29} (2000) 149--157.

\bibitem{Soicher}
L. H. Soicher, Computation of partial spreads web-page, \\ \url{http://www.maths.qmul.ac.uk/~lsoicher/partialspreads/}

\bibitem{WW1996} H. N. Ward and J. A. Wood, Characters and the equivalence of codes, {\it J. Combin. Theory Ser. A}, {\bf 73} (1996) 348--352.


\end{thebibliography}
\end{document}